%% file: stacs.tex
\documentclass[a4paper,UKenglish,cleveref,autoref,thm-restate]{lipics-v2021}
\title{On Rational Recursive Sequences}

\author{Lorenzo Clemente}{University of Warsaw, Poland}{clementelorenzo@gmail.com}{https://orcid.org/0000-0003-0578-9103}{partially supported by the ERC grant INFSYS, agreement no. 950398.}
\author{Maria Donten-Bury}{University of Warsaw, Poland}{m.donten@mimuw.edu.pl}{https://orcid.org/0000-0002-2138-4109}{partially supported by the National Science Center, Poland, project 2017/26/E/ST1/00231.}
\author{Filip Mazowiecki}{University of Warsaw, Poland}{f.mazowiecki@mimuw.edu.pl}{}{partially supported by the ERC grant INFSYS, agreement no. 950398.}
\author{Micha\l{} Pilipczuk}{University of Warsaw, Poland}{michal.pilipczuk@gmail.com}{}{partially supported by the ERC grant BOBR, agreement no. 948057.}%TODO mandatory, please use full name; only 1 author per \author macro; first two parameters are mandatory, other parameters can be empty. Please provide at least the name of the affiliation and the country. The full address is optional. Use additional curly braces to indicate the correct name splitting when the last name consists of multiple name parts.
\authorrunning{L. Clemente, M. Donten-Bury, F. Mazowiecki and M. Pilipczuk} %TODO mandatory. First: Use abbreviated first/middle names. Second (only in severe cases): Use first author plus 'et al.'

\Copyright{Lorenzo Clemente, Maria Donten-Bury, Filip Mazowiecki and Micha\l{} Pilipczuk} %TODO mandatory, please use full first names. LIPIcs license is "CC-BY";  http://creativecommons.org/licenses/by/3.0/

 %TODO mandatory: Please choose ACM 2012 classifications from https://dl.acm.org/ccs/ccs_flat.cfm 

 %TODO mandatory; please add comma-separated list of keywords

\category{} %optional, e.g. invited paper

\relatedversion{} %optional, e.g. full version hosted on arXiv, HAL, or other respository/website
\acknowledgements{I want to thank \dots}%optional

\nolinenumbers %uncomment to disable line numbering

\EventEditors{John Q. Open and Joan R. Access}
\EventNoEds{2}
\EventLongTitle{42nd Conference on Very Important Topics (CVIT 2016)}
\EventShortTitle{CVIT 2016}
\EventAcronym{CVIT}
\EventYear{2016}
\EventDate{December 24--27, 2016}
\EventLocation{Little Whinging, United Kingdom}
\EventLogo{}
\SeriesVolume{42}
\ArticleNo{23}
\usepackage{amsmath}
\usepackage{amssymb}
\usepackage{amsthm}
\usepackage{mathabx}

\usepackage{graphicx}%,algorithmic,algorithm}%, verbatim}
\usepackage{tikz}\usetikzlibrary{shapes}
\usepackage{todonotes}
\usepackage{hyperref}
\usepackage{cleveref}
\usepackage{xspace}
\usepackage{stmaryrd}
\usepackage{mathtools}

\usepackage{comment}

\usepackage{xcolor}
\hypersetup{
    colorlinks,
    linkcolor={red!50!black},
    citecolor={blue!50!black},
    urlcolor={blue!80!black}
}

\input{macros}

\begin{document}
\maketitle

\begin{abstract}
We study the class of rational recursive sequences (ratrec) over the rational numbers.
A ratrec sequence is defined via a system of sequences using mutually recursive equations of depth 1,
where the next values are computed as rational functions of the previous values.
An alternative class is that of simple ratrec sequences,
where one uses a single recursive equation, however of depth $k$: the next value is defined as a rational function of $k$ previous values.

We conjecture that the classes ratrec and simple ratrec coincide. The main contribution of this paper is a proof of a variant of this conjecture where the initial conditions are treated symbolically, using a formal variable per sequence, while the sequences themselves consist of rational functions over those variables. While the initial conjecture does not follow from this variant, we hope that the introduced algebraic techniques may eventually be helpful in resolving the problem.

The class ratrec strictly generalises a well-known class of polynomial recursive sequences (polyrec). These are
defined like ratrec, but using polynomial functions instead of rational ones.
One can observe that if our conjecture is true and effective, then we can improve the complexities of the zeroness and the equivalence problems for polyrec sequences. Currently, the only known upper bound is Ackermanian, which follows from results on polynomial automata.
We complement this observation by proving a~\PSPACE lower bound for both problems for polyrec. Our lower bound construction also implies that the Skolem problem is \PSPACE-hard for the polyrec class.
% The upper and the lower bound for zeroness and equivalence also hold for polynomial automata over a unary alphabet, since this class is known to be equivalent to polyrec.
% We conjecture that if ratrec was introduced in the early 90s Ariane 5 would now be in space, safely computing the 5000th element of Catalan numbers.
\end{abstract}

\newpage

\section{Introduction}
\label{sec:introduction}
\input{introduction}

\section{Preliminaries}
\label{sec:preliminaries}
\input{preliminaries}

\section{Rational recursive sequences}
\label{sec:rationalrec}
\input{rationalrec}
\section{Transcendence degrees}
\label{sec:transcendence}
\input{transcendence}

% \section{Zeroness for polynomial automata over unary alphabets: upper bound}
% \label{sec:upper}
% \input{upper}

\section{Conclusion}
\label{sec:conclusions}
\input{conclusions}
\bibliography{references}

\appendix
% \label{sec:appendix}
% \input{appendix}

\section{Zeroness for polyrec is PSPACE-hard}
\label{sec:lower}
\input{lower}

% that's all folks
\end{document}

%% file: macros.tex
\newcommand*{\eg}{e.g.\@\xspace}
\newcommand*{\ie}{i.e.\@\xspace}

\newcommand{\bu}{\mathbf{u}}
\newcommand{\bv}{\mathbf{v}}

\newcommand{\bx}{\mathbf{x}}
\newcommand{\bc}{\mathbf{c}}
\newcommand{\bd}{\mathbf{d}}
\newcommand{\bff}{\mathbf{f}}
\newcommand{\bF}{\mathbf{F}}

\newcommand{\cA}{\mathcal{A}}
\newcommand{\cB}{\mathcal{B}}

\newcommand{\bbN}{\mathbb{N}}
\newcommand{\N}{\bbN}

\newcommand{\bbQ}{\mathbb{Q}}
\newcommand{\bbR}{\mathbb{R}}
\newcommand{\bbD}{\mathbb{D}}

\newcommand{\bbE}{\mathbb{E}}
\newcommand{\bbF}{\mathbb{F}}
\newcommand{\bbG}{\mathbb{G}}

\renewcommand{\leq}{\leqslant}
\renewcommand{\geq}{\geqslant}
\renewcommand{\le}{\leqslant}
\renewcommand{\ge}{\geqslant}

\newcommand{\set}[1]{\{ #1 \}}
\newcommand{\sem}[1]{\left\llbracket #1 \right\rrbracket}

\newcommand{\trdeg}{\operatorname{tr\,deg}}

\newcommand{\PTIME}{\textsf{PTIME}\xspace}

\newcommand{\NPTIME}{\textsf{NP}\xspace}

\newcommand{\PSPACE}{\textsf{PSPACE}\xspace}

\newcommand{\NC}[1]{$\textsf{NC}^{\text{#1}}$\xspace}

\newcommand{\QBF}{\textsf{QBF}\xspace}

\newcommand{\michal}[1]{\todo[inline, color=blue!30]{{\bf Mi:} #1}}

\newcommand{\lorenzo}[1]{\todo[inline, color=orange!30]{{\bf L:} #1}}

\crefname{equation}{equation}{equations}
\crefname{lemma}{Lemma}{Lemmas}
\crefname{theorem}{Theorem}{Theorems}
\crefname{claim}{Claim}{Claims}
\crefname{section}{Section}{Sections}
\crefname{definition}{Definition}{Definitions}
\crefname{proposition}{Proposition}{Propositions}
\crefname{corollary}{Corollary}{Corollaries}

%% file: introduction.tex
The topic of this paper are recursively defined sequences of rational numbers $\bbN \to \bbQ$.
There are two natural ways to define such sequences.
In a \emph{simple recursion of depth $k$} one fixes $k$ initial values and defines the next value as a function of the previous $k$ values.
This is how the Fibonacci sequence is usually defined (with $k = 2$):
$f_0 = 0$, $f_1 = 1$, and $f_{n+2} = f_{n+1} + f_n$.
%We call this a \emph{single deep equation}.
%
In a \emph{mutual recursion of width~$k$} one defines a system of $k$ sequences such that every sequence has its initial value and the update function can access the immediately previous value of all $k$ sequences, but no older value.
For example, we can define $a_n = n^2$ with an extra sequence $b_n = n$ as follows:
$a_0=b_0=0$ and $a_{n+1} = a_n + 2b_n + 1$, $b_{n+1} = b_n +1$.
%We call this a \emph{system of shallow equations}.
Both styles allow to define various classes of sequences depending on what operations are allowed in the equations,
and in general mutual recursion of width $k$ can simulate simple recursion of depth $k$
(by adding sufficiently many auxiliary sequences).
%
%To avoid ambiguity we will write whether we consider the class defined in the former way with \emph{one equation}
%or in the latter way with a \emph{system of equations}.
% \lorenzo{We should emphasize that in the first approach we use ``complete induction'', while in the second one we use ``induction''}
% \filip{I don't know what that means}
% \lorenzo{We should say precisely what we mean by sequences, i.e., number sequences ``$\bbN \to \bbQ$''.}
% \filip{I think it makes the intro go towards technicalities but you can add it if you disagree}
% \lorenzo{Shall we introduce some name for the two styles of definitions above?
% Like: type A/type B or something else like simple linrec / linrec and simple polyrec / polyrec.
% I found confusing in the following to follow informal phrases such as ``with one sequence'' / ``with a system of sequences''.}
% \filip{I haven't found a good solution to this. I agree currently it's not amazing but if we have type A/B then I'd rather they have some meaningful names. I don't like introducing technical terms in the introduction (by the way I removed ratrec later :P )}

One of the most well-known classes of sequences is the class of \emph{linear recursive sequences}, %over $\bbQ$,
which is obtained by allowing the update function to use addition and multiplication with constants.
These are usually defined with a simple recursion, like in the Fibonacci example,
but in fact, as a consequence of the Cayley-Hamilton theorem, one obtains the same class when using mutual recursion~\cite[Lemma 1.1]{HalavaHarjuHirvensaloKarhumaki:techrep:2005}.
In particular, all the example sequences $f_n$, $a_n$ and $b_n$ are linear recursive.

Another natural class of sequences are the \emph{polynomial recursive sequences} (\emph{polyrec}), %over $\bbQ$
which are defined with mutual recursion and updates from the ring of polynomial functions $\bbQ[x_1, \dots, x_k]$.
%The recursion is defined by polynomials, where variables correspond to previous values.
An example sequence from this class is $c_n = n!$, where one can use the already defined sequence $b_n$ and define $c_0 = 1$ and $c_{n+1} = c_n \cdot (b_n + 1)$. To see the polynomials behind this definition, let $x$ and $y$ be variables corresponding to $b_n$ and $c_n$, respectively. The polynomial to define $b_{n+1}$ is $P_b(x,y) = x + 1$, and the polynomial to define $c_{n+1}$ is $P_c(x,y) = y(x+1)$.
The class of \emph{simple polynomial recursive sequences} is obtained by using polynomial updates and a simple recursion (instead of mutual recursion) 
and it is known to be strictly included in the class of all polyrec sequences.
In particular, the sequence $c_n$ is polyrec but not simple polyrec~\cite[Theorem 3.1]{cadilhac2021polynomial}.

The definition via mutual recursion appears in the area of control theory
(under the name \emph{implicit representation} of the space of states),
and, in computer science, in the context of \emph{weighted automata} over $\bbQ$.
Such automata output a rational number for every word over a finite alphabet $\Sigma$,
and they are defined by linear updates~\cite{DrosteHWA09}.
Linear recursive sequences are thus equivalent to weighted automata with a $1$-letter alphabet $\Sigma = \set a$~\cite{BarloyFLM20}.
Similarly, polyrec sequences are equivalent to \emph{polynomial automata}~\cite{BenediktDuffSharadWorrell:PolyAut:LICS:2017} (also known as \emph{cost-register automata}~\cite{AlurDDRY13}) with a $1$-letter alphabet~\cite{cadilhac2021polynomial}.

We are interested in two classical decision problems for such automata.
\emph{Equivalence}: Given two automata $\cA$ and $\cB$ do they output the same number for every word, and \emph{zeroness}: Does the input automaton $\cA$ output $0$ for every word.
These problems are well-known to be efficiently equivalent to each other:
Zeroness is clearly a special case of equivalence (just take $\cB$ to output zero for every word),
and equivalence of $\cA, \cB$ reduces to zeroness of the difference automaton $\cA - \cB$ with the expected semantics.
Therefore, we will consider only the zeroness problem.
From the seminal work of Sch\"{u}tzenberger on minimisation of weighted automata
it follows that the zeroness problem for weighted automata is in \PTIME~\cite{Schutzenberger:IC:1961}
(in fact even in \NC 2~\cite{Tzeng:IPL:1996}).
%\lorenzo{Much more than that, he even shows that linear weighted automata over a field can effectively be minimised to a unique normal form, and thus zeroness reduces to minimisation + equality with the trivial zero automaton.}
%\filip{Yes. I think it would be nice to write about it if we would discuss that similarly LRS have a minimal degree. But that's becoming too off topic :)}
For polynomial automata over a binary alphabet, zeroness is known to be Ackermann-complete~\cite{BenediktDuffSharadWorrell:PolyAut:LICS:2017}. 
Using the connection between sequences and automata one immediately obtains \NC 2 and Ackermann upper bounds for the zeroness problem of linear recursive sequences, resp., polyrec sequences.

Let us take a closer look at the zeroness problem for recursive sequences,
\ie, given a sequence $u_n$ is it the case that $u_n = 0$ for all $n \in \N$?
The zeroness problem is a fundamental problem for number sequences.
It is a basic building block in computer algebra, e.g.,
in proving identities involving recursively defined sequences.
It is also important from a theoretical point of view as a yardstick of the well-behavedness of classes of number sequences,
i.e., interesting classes of sequences should at least have a decidable zeroness problem.
The difficulty of solving the zeroness problem in general depends on how the sequence is presented.
If the sequence is defined with a simple recursion of depth $k$ such as $u_{n+k} = f(u_{n+k-1}, \dots, u_n)$,
then zeroness trivially reduces to checking that the first $k$ values are $0$
and that the recursive update $f$ is well-defined and needs to output $0$ when the previous values are $0$,
\ie, $f(0, \dots, 0) = 0$.
However, this simple reasoning is flawed in the case of mutual recursion,
because the auxiliary sequences employed in the mutual recursion need not be zero.
However, for linear recursive sequences the zeroness problem is easily solved
even in the case of mutual recursion,
because the reduction to simple recursion \cite[Lemma 1.1]{HalavaHarjuHirvensaloKarhumaki:techrep:2005}
implies that $a_n$ is zero if, and only if, its first $k+1$ values $a_0 = \cdots = a_k$ are zero.
For polyrec sequences we cannot apply this argument
since mutual recursion cannot be simulated by simple recursion in the case of polynomial updates.
%\lorenzo{We should state a precise upper bound :)}
%\filip{Trivial-complete? :) I'd rather conclude that it's not an interesting problem for LRS}
%\lorenzo{Shall we say somewhere that the nonemptiness problem, i.e., is there an input yielding zero, is undecidable for weighted automata?}
%\filip{Right. I think it would be best to write about it in the related work part because there we will discuss the corresponding Skolem problem.}

% \bigskip

\subparagraph*{Our results}
In this paper we introduce the class of \emph{rational recursive sequences} (\emph{ratrec}). %over $\bbQ$.
This class is defined with mutual recursion and updates from the field of rational functions $\bbQ(x_1, \dots, x_k)$.
For example, the Catalan numbers $C_{n+1} = \frac{2(2n+1)}{n+2} C_n$
can be defined using $b_n$ as an auxiliary sequence. Namely, $C_0 = 1$ and $C_{n+1} = \frac{2(2b_n + 1)}{b_n+2}C_n$, where the rational function used to define $C_{n+1}$ is $R(x,y) = \frac{2(2x+1)}{x+2}y$.
By definition, the class of polyrec sequences is included in the class of ratrec sequences,
and in fact the inclusion is strict as witnessed by the fact that the Catalan numbers $C_n$ are not polynomialy recursive~\cite[Corollary~4.1]{cadilhac2021polynomial}.
Moreover, ratrec sequences also include the well-known and wide-spread \emph{P-recursive sequences}%
\footnote{
Sometimes P-recursive sequences are also called \emph{holonomic sequences},
due to a connection with holonomic generating functions.
}~\cite{KauersP11},
which according to a 2005 estimate comprise at least 25\% of the OEIS archive \cite{Salvy:ISAAC:2005}.

% \lorenzo{The official name is P-recursive. We can add a footnote explaining that sometimes they are also called holonomic sequences since:
% 1) the generating function of a P-recursive sequence is a holonomic function,
% and 2) the coefficients of the power series expansion of a holonomic function form a P-recursive sequence.
% A function is \emph{holonomic} if it is an element of a holonomic module of smooth functions.
% An equivalent definition is with differential equations:
% A function is \emph{D-finite} if its set of derivatives spans a finite-dimensional vector space.}
% \filip{This sounds good.}

A natural question is whether the class of ratrec sequences semantically collapes to the class of \emph{simple rational recursive sequences} obtained by adopting simple recursion.
Unlike in the case of polynomial updates, we conjecture that for rational updates we do have such a~collapse.

\begin{restatable}{conjecture}{thmMain}
    \label{thm:main}
    The class of rational recursive sequences coincides with the class of simple rational recursive sequences.
\end{restatable}

To see the power of ratrec sequences recall that $c_n = n!$ is not a simple polyrec sequence.
However, when in the recursion we allow rational functions, then $c_n$ can be defined with a simple recursion, namely: $c_{n+2} = \frac{(c_{n+1})^2}{c_n} + c_{n+1}$.
Thus $c_n$ is simple ratrec.

% \michal{We need to clearly state here what we prove and how this relates to \cref{thm:main}.}

We introduce a technique towards proving~\cref{thm:main}, which comes from commutative algebra. Instead of looking at the elements of a ratrec sequence as numbers in the field of rationals $\bbQ$,
we symbolically view them as elements of the field of rational functions $\bbQ(x_1, \dots, x_k)$. More precisely, we assume that the sequences $\bF^{(1)}, \ldots, \bF^{(k)}$ are initialised by setting $F_0^{(i)} = x_i$ for all $i\in \{1,\ldots,k\}$; then, a system of recursive equations governed by rational functions defines further entries of the sequences. Thus, the recursive definition will output elements in $\bbQ(x_1, \dots, x_k)$ rather than $\bbQ$. Intuitively, this corresponds to treating the initial conditions of a system of ratrec sequences symbolically, rather than instantiating them with actual rational values.

Informally speaking, we prove \cref{thm:main} for symbolic ratrec sequences, as explained above.
Here is a semi-formal statement of our main result, see \cref{thm:simple:ratrec} for a formalization.

\begin{theorem}\label{thm:main-informal}
    The class of rational recursive sequences over $\bbQ(x_1,\ldots,x_k)$, with the system initialised by $F_0^{(i)} = x_i$, coincides with the class of simple rational recursive sequences.
\end{theorem}
%L: I find this comment on subsitutiono unclear without extra context
%formulas ready for substituting the initial values.

The proof proceeds as follows.
From the functions defining the ratrec system we build a sequence of field extensions
\[\bbQ \subseteq \bbF_0 \subseteq \bbF_1 \subseteq \bbF_2 \subseteq \ldots \subseteq \bbQ(x_1,\ldots,x_k)\]
and translate the problem of belonging to the class of simple ratrec sequences to the question of whether this sequence of field extensions eventually stabilises.
In order to estimate at which level the stabilisation occurs we use certain results on basic algorithms for rational function fields~\cite{Muller-Quade:Steinwandt:JSC:1999}.
%
%An interesting point is that the natural algebraic object corresponding to a simple ratrec sequence is the tower of field extension,
%as opposed to the tower of ideals, in which case one could simply apply the Hilbert basis theorem.
We believe that this technique could be extended to prove \cref{thm:main}, but we also show an example why our current results are not strong enough.

Note that if~\cref{thm:main} is moreover efficient, it gives a simple algorithm to check zeroness for polyrec. Indeed, since polyrec is a particular case of ratrec, then once a sequence is expressed as a simple ratrec it suffices to check whether the first elements of the sequence are 0. This would improve the
Ackermann upper bound inherited from polynomial automata from \cite{BenediktDuffSharadWorrell:PolyAut:LICS:2017}.
This suggests that for polyrec sequences the natural object of study are rational function fields,
which are of more algebraic nature and could provide better complexity bounds
than the order-theoretic techniques based on sequences of polynomial ideals and Hilbert's finite basis theorem \cite{BenediktDuffSharadWorrell:PolyAut:LICS:2017}.

Our final result is a complexity lower bound for the zeroness problem of polyrec sequences.

\begin{restatable}{theorem}{thmZeronessLowerBound}
    \label{thm:lower-bound}
    The zeroness problem for polynomial recursive sequences is \PSPACE-hard.
\end{restatable}

As far as we know, prior to this work nothing was known about the complexity of zeroness for polyrec sequences, except for the Ackermann upper bound following from polynomial automata~\cite{BenediktDuffSharadWorrell:PolyAut:LICS:2017}.
%
%This also implies that the exponential blow up in the translation from a system of equations to a single equation from \cref{thm:main} is unavoidable,
%unless $\coRPTIME = \PSPACE$\mlorenzo{right?}.\filip{No, the blow up is unconditional. If you look at the construction then every true formula will define a sequence that is $0$ for exponentially many steps and then at the very end it becomes $1$.}
The lower bound is proved by reducing from the \QBF validity problem.
% and as a consequence of the construction it follows that the exponential blow-up of the recursion depth in the translation from \cref{thm:main} is unavoidable.

Given \cref{thm:main} it seems natural to investigate the zeroness problem for ratrec sequences. % not just for polyrec sequences.
The issue is that it is not clear what would be the input for such a decision problem. Recall that to define ratrec sequences we allow for rational functions in the recursion, which means that we have to deal with division in order to compute the elements of the sequences. Then either one would require that the input sequence comes with a promise that all elements are well-defined and no division by $0$ occurs; or one would need to verify whether division by $0$ occurs in the input sequence. We find the former solution unnatural,
and the latter is at least as hard as the so-called Skolem problem (c.f.~below),
which is not known to be decidable even for linear recursive sequences.

% \bigskip

\subparagraph*{Related work}
The zeroness problem has been extensively studied.
In the field of automata theory,
we can mention applications to the equivalence problem of several classes of automata and grammars,
starting from weighted finite automata \cite{Schutzenberger:IC:1961}
and polynomial automata \cite{BenediktDuffSharadWorrell:PolyAut:LICS:2017} already mentioned above,
and including context-free grammars \cite{ChomskySchutzenberger:Algebraic:1963},
multiplicity equivalence of finite automata \cite{Tzeng:IPL:1996}
and multitape finite automata \cite{HarjuKarhumaki:TCS:1991,Worrell:ICALP:2013},
unambiguous context-free grammars \cite[Theorem 5.5]{SalomaaSoittola:Book:PowerSeries:1978}
(c.f.~\cite{ForejtJancarKieferWorrell:IC:2014,Clemente:EPTCS:2020} for a \PSPACE upper bound),
polynomial grammars (which generalise polynomial automata) \cite[Chapter 11]{BojanczykCzerwinski:Toolbox:2018},
deterministic top-down tree-to-string transducers \cite{SeidlManethKemper:JACM:2018},
MSO transductions on unordered forests \cite{BoiretPiorkowskiSchmude:FSTTCS:2018,Bojanczyk:SIGLOG:2019},
MSO transductions of bounded treewidth under a certain equivalence relation \cite{BojanczykSchmude:MFCS:2020},
Parikh automata \cite{BostanCarayolKoechlinNicaud:ICALP:2020},
and unambiguous register automata \cite{BarloyClemente:STACS21}.
By replacing (pointwise) multiplication with convolution in the definition of polyrec sequences we obtain the so-called \emph{convolution recursive sequences},
for which the zeroness problem can be solved in \PSPACE \cite[Theorem 4]{Clemente:EPTCS:2020}.
%\filip{better than?}
%
% While convolution-recursive sequences have algebraic generating functions,
% and thus the underlying number sequences are P-recursive,
% it is not clear how to efficiently (i.e., in \PTIME) convert a system of algebraic equations for the generating functions into a P-recursive description (as a single deep equation) for the underlying nuumber sequence.

The zeroness problem of D-finite \cite{Zeilberger:JCAM:1990} and, more generally,
D-algebraic power series \cite{DenefLipshitz:JSL:1989,Hoeven:AAECC:2019} is known to be decidable,
but its computational complexity has not been investigated.

A natural problem related to the zeroness problem is the so-called \emph{Skolem problem},
which asks whether a given sequence $a_n$ has a zero,
i.e., whether for some $n$ we have $a_n = 0$.
As a corollary of the constructions used to prove \cref{thm:lower-bound},
it follows that the Skolem problem for polyrec sequences is \PSPACE-hard.
Only \NPTIME-hardness was formerly known,
and already for linear recursive sequences~\cite[Corollary 2.1]{BlondelPortier:LAA:2002}.
Decidability of the Skolem problem for linear recursive sequences is a long-standing open problem
(c.f.~the survey paper \cite{OuaknineWorrell:SIGLOG:2015}).
It is interesting to notice that those lower bounds are obtained already on the fixed field with two elements $\set{0, 1}$,
and are thus of a combinatorial rather than numerical nature.
The Skolem problem for weighted automata over $\bbQ$ (that generalise linear recursive sequences) is undecidable~\cite{paz71}.
%
% Moreover, on fixed finite fields the values of a ratrec sequence are eventually periodic with an exponential period,
% and thus both the zeroness and Skolem problems can be decided by inspecting a prefix of the sequence of exponential length, which can be done in \PSPACE for every fixed finite field.

%Kemper~\cite{Kemper:1993} and Sweedler~\cite{Sweedler:AAECC19:1993}.

%Full proofs can be found in an accompaying technical \cref{sec:appendix}.

%% file: preliminaries.tex
By $\bbN$ we denote the set of nonnegative integers.
We denote an arbitrary field by $\bbF$,
and we use $0$ and $1$ to denote the zero, resp., one elements thereof.
Example fields of interest in this paper are: rationals $\bbQ$; and the two-element field $\bbF_2$.
A {\em{sequence}} over a \emph{domain} $\bbD$ is a function $u \colon \bbN \to \bbD$. The sequences considered in this work are over domains that have a field structure, like rationals $\bbQ$.
%We use the notation $\langle u_n \rangle_{n \in \bbN}$ for elements of sequences, where $u_n=u(n)$.
We use bold-face letters as a short-hand for sequences, e.g., $\bu = \langle u_n
\rangle_{n \in \bbN}$.

In this paper we work with multivariate polynomials and rational functions.
The \emph{(combined) degree} of a monomial $x_1^{d_1} \cdots x_k^{d_k}$ is $d_1 + \cdots + d_k$ and the degree of a polynomial $P \in \bbQ(x_1, \dots, x_k)$, written $\deg P$, is the maximum degree of monomials appearing in it.
A {\em{rational function}} is a formal fraction of two polynomials, where the denominator is required to be non-zero. The degree of a rational function is the maximum of the degrees of the numerator and the denominator.
Recall that for any field $\bbF$ and a set of variables $x_1,\ldots,x_n$, polynomials over $x_1,\ldots,x_n$ form the ring $\bbF[x_1,\ldots,x_n]$, while rational functions over $x_1,\ldots,x_n$ form the field $\bbF(x_1,\ldots,x_n)$. We also write $\bbF[\bx]$ and $\bbF(\bx)$, where $\bx = (x_1,\ldots,x_n)$.

% The \emph{value} of a rational function $P/Q$ at a point $\bx = (x_1,\ldots,x_k)$ is the quotient of values $P(\bx)$ and $Q(\bx)$. It is well-defined if $Q(\bx) \neq 0$, and if $Q(\bx) = 0$ we say that $P/Q$ has a \emph{singularity} at $\bx$. However, it may happen that at a singular point $P/Q$ still can be assigned a value in a unique reasonable way. Namely, if $P(\bx) = Q(\bx) = 0$ and the limit $\lim\limits_{x\rightarrow \bx} P(x)/Q(x)$ exists then the singularity at $\bx$ is \emph{removable} and one can define $(P/Q)(\bx) = \lim\limits_{x\rightarrow \bx} P(x)/Q(x)$. Note that this can be done if a field allows to take limits, e.g., over $\mathbb{Q}$.
% 
% The value of $P/Q$ at $\bx$ where we have a removable singularity can be computed effectively. One may restrict to a line $L$ through $\bx$ and an arbitrarily chosen point $\mathbf{y}$ such that $Q(\mathbf{y}) \neq 0$. By substituting the equations of $L$ to $P$ and $Q$ one obtains univariate polynomials. The value of their quotient at $\bx$ exists since we assume that the singularity is removable and it can be computed e.g. using the de l'Hospital's rule, formally differentiating both polynomials unless the substitution of $\bx$ gives non-zero value of the denominator.

The computational aspects of multivariate polynomials, in particular their representation on input to algorithms, are explained in \cref{sec:lower}, as they will be of no concern in \cref{sec:rationalrec,sec:transcendence}.

% The \emph{height} of a rational function over $\bbQ$ is the maximum absolute value of a coefficient occurring in its definition.
% 
% \michal{What is the {\em{definition}} of a rational function? Do we mean expansions of the numerator and the denominator into sums of monomials?}
% 
% For computational purposes, the numerical coefficients appearing in polynomials and rational functions are encoded in binary, while the degrees of variables are assumed to be encoded in unary (thus we adopt a non-sparse representation).
% \filip{With the circuit encoding can we also assume that degrees are in binary? (not that it matters but it will be more uniform)}

%In order to ignore $\log$-factors in arithmetic computations we use the \emph{soft-O} notation $\softO{g(n)} := O(g(n) \cdot \log g(n))$.

%% file: rationalrec.tex
We start with the central definitions, which were already discussed in Section~\ref{sec:introduction}.

\begin{definition}\label{def:ratrec}\label{def:polyrec}
    A sequence $\bu^{(1)}$ over a field $\bbF$ is {\em{rationally recursive}} (or {\em{ratrec}} for short) of \emph{dimension} $k$ and \emph{degree} $D$ if there exist auxiliary sequences $\bu^{(2)},\ldots,\bu^{(k)}$ over $\bbF$ and rational functions $P_1,\ldots,P_k\in \bbF(x_1,\ldots,x_k)$ of degree at most $D$ such that for all $n\in \bbN$, we have
    \begin{align}\label{def:prs}
        \left\{\begin{array}{lcl} 
            u^{(1)}_{n+1}   &=& P_1(u^{(1)}_n,\ldots,u^{(k)}_n), \\
                        &\vdots& \\
            u^{(k)}_{n+1}   &=& P_k(u^{(1)}_n,\ldots,u^{(k)}_n).
        \end{array}\right.
    \end{align}
    A sequence $\bu$ over a field $\bbF$ is {\em{polynomially recursive}} (or {\em{polyrec}} for short)
    if it satisfies the same definition above, where $P_1,\ldots,P_k$ are taken as polynomials in $\bbF[x_1,\ldots,x_k]$.
    We refer to $(\bu^{(1)},\ldots,\bu^{(k)})$ as the system defining $\bu^{(1)}$.
\end{definition}

In what follows we assume that whenever $\bu$ is a ratrec sequence, say defined by a system $(\bu=\bu^{(1)},\ldots,\bu^{(k)})$, for all $n\in \bbN$ all the right hand sides of equations~\eqref{def:prs} are well-defined, that is, no denominator of any rational expression contained in the right hand side is zero.

%all singularities of rational functions which occur in the process of determining the values of a ratrec sequence for chosen initial conditions are removable.

For instance, the sequence of \emph{Catalan numbers} $C_n = \frac 1 {n+1} \cdot {2n \choose n}$ is ratrec.
This can be seen in several ways. For example, they satisfy the recurrence $C_{n+1} = \frac {2(2n + 1)}{n+2} \cdot C_n$,
giving rise to the following ratrec system:
\begin{align*}
    \left\{\begin{array}{ll}
    u_{n+1} &= \frac {2(2v_n + 1)}{v_n+2} \cdot u_n, \\
    v_{n+1} &= v_n + 1.
    \end{array}\right.
\end{align*}

More generally, any P-recursive sequence $a_n$ is ratrec.
A sequence is \emph{P-recursive} \cite[Sec.~6.4]{StanleyFomin:EU:CUP:2001} if it satisfies a single recursion of the form
\begin{align}\label{def:P-rec}
    P_0(n) \cdot a_n + P_1(n) \cdot a_{n+1} + \cdots + P_d(n) \cdot a_{n+d} = 0,
\end{align}
for every $n$ large enough,
where $P_0, \dots, P_d \in \bbQ[n]$ are polynomials of the index variable $n$.
This is readily transformed into the ratrec system

\begin{align*}
    \left\{\begin{array}{lcl}
        u^{(d)}_{n+1}
            &=& - \frac {P_0(v_n)} {P_d(v_n)} \cdot u^{(0)}_n - \cdots - \frac {P_{d-1}(v_n)} {P_d(v_n)} \cdot u^{(d-1)}_n, \\
        u^{(d-1)}_{n+1},
            &=& u^{(d)}_{n}, \\
            &\vdots& \\
        u^{(0)}_{n+1} &=& u^{(1)}_{n}, \\
        v_{n+1} &=& v_n + 1.
    \end{array}\right.
\end{align*}

Assuming $v_0 = 0$ and $u^{(0)}_0 = a_0, \dots, u^{(d)}_0 = a_d$,
it is immediate to verify $v_n = n$ and $u^{(0)}_n = a_n, \dots, u^{(d)}_n = a_{n+d}$ for every $n \in \bbN$.

%\lorenzo{We need to address what happens when $P_d(n)$ is $0$.}

% \filip{I would consider using the name holonomic sequences. I know it's incorrect but I think most people heard about this class from Joel+Ben who write holonomic. We could explain in the introduction that this is equivalent.}

The family of ratrec sequences strictly includes both {P-recursive} sequences and polyrec sequences.
As an example consider the sequence $u_n = 2^{2^n} + C_n$.
On the one hand, this sequence is certainly ratrec because it is the sum of a polyrec and a {P-recursive} sequence (which are ratrec)
and ratrec sequences are closed under sum.
On the other hand, $u_n$ is not P-recursive since it grows asymptotically faster than any P-recursive sequence
(every {P-recursive} sequence is in $O((n!)^\gamma)$ for some constant $\gamma\in\bbR$ \cite[Proposition 3.11]{Lipshitz:D-finite:JA:1989}). Further, $u_n$ is also not polyrec, because $2^{2^n}$ is polyrec, $C_n$ is not \cite[Corollary~4.1]{cadilhac2021polynomial}, and polyrec sequences are closed under~sum and subtraction.

In \cite[Theorem~7.1]{cadilhac2021polynomial}, the following property of ratrec sequences is proved: if $\bu$ is ratrec, then there exists $m\in \N$ and a {\em{cancelling polynomial}} $P\in \bbQ[y_0,\ldots,y_m]$, that is, a non-zero polynomial such that
$$P(u_n,u_{n+1},\ldots,u_{n+m})=0\qquad\textrm{for all }n\in \bbN.$$
% The proof of~\cite[Theorem~11]{CadilhacMPPS20} in fact operates on a field of rational functions over $\bbQ$ and thus readily works for ratrec sequences as well. Thus we obtain:

\begin{theorem}[Theorem~7.1 in \cite{cadilhac2021polynomial}]\label{thm:ratrec-cancelling}
 Every ratrec sequence admits a cancelling polynomial.
\end{theorem}

In~\cite[Theorem~5.3]{cadilhac2021polynomial} it is shown that the sequence $u_n=n^n$ has no cancelling polynomial, and hence is not polyrec and not ratrec.

%% file: transcendence.tex
%In this section we fix a ratrec system as in \cref{def:ratrec} of dimension $k$ and degree $D$.

% \michal{I do not find the motivation above convincing. I much more like the following narrative. There are two natural ways to define poly-rec and ratrec: via a system of equations of depth $1$, or via one equation of possibly larger depth. In~\cite{CadilhacMPPS20} it was shown that these two ways give different classes for poly-rec, but here we show that they give the same classes for ratrec.}

In this section we consider ratrec sequences as in \cref{def:ratrec} over the field $\bbQ(\bx)$.
Let $(\bF^{(1)},\ldots,\bF^{(k)})$ be a system defining $\bF^{(1)}$.
In this section we will consider sequences with the following fixed initial conditions: $F_0^{(i)}= x_i$. Note that this technical assumption is important, in particular we cannot initialise $F_0^{(i)}$ with elements in $\bbQ$. (If we could, this class would generalise ratrec over the field $\bbQ$.)
%
% \begin{align}
%     \label{eq:ratrec}
%     F_{n+1}^{(i)} = P_i(F_n^{(1)}, \ldots, F_n^{(k)}), \quad \text{for all } n \in \bbN,
% \end{align}
% where $P_i \in \bbQ(\bx)$. Thus, each sequence $\bF^{(i)}$ is a ratrec sequence over the field $\bbQ(\bx)$.
% , and is mapped to the sequence $u^{(i)}_n$ by an operation that substitutes variables $x_1,\ldots,x_k$ with the initial condition $u^{(1)}_0,\ldots,u^{(k)}_0$. 
% Note that in the recursive definition of the rational functions $F_n^{(i)}$ it never happens that a denominator of a fraction turns out to be a zero polynomial. Indeed, if this was the case, then the same would happen in the recursive definition of the sequences $u^{(i)}_n$, contradicting the assumption that the considered system defines the sequence $\bu^{(1)}$.

\cref{thm:simple:ratrec} below formalises \cref{thm:main-informal} and is the main result of this paper. 
In essence, we show that a ratrec definition over $\bbQ(\bx)$ can be translated to a simple ratrec over $\bbQ(\bx)$ with a polynomial recursion depth.
We hope that this insight might lead towards a resolution of \cref{thm:main}.

\begin{theorem}\label{thm:simple:ratrec}
    Let $\bF^{(1)}$ be a ratrec sequence over the field $\bbQ(\bx)$, defined by a system $(\bF^{(1)},\ldots,\bF^{(k)})$, with the initial conditions: $F_0^{(i)}= x_i$ for $i=1,\ldots,k$.
    Then there exists a rational function 
    $R \in \bbQ(y_0,\ldots,y_m)$ such that
    %for some $m \in \bbN$ such that, for all $n \in \bbN$,
    %
    \begin{align*}
        F^{(1)}_{n+m+1} = R(F^{(1)}_n,F^{(1)}_{n+1},\ldots,F^{(1)}_{n+m}), \text{ for all } n \in \bbN.
    \end{align*}
    Moreover, if $F_n^{(1)}$ is of dimension $k$ and degree $D$,
    then $m$ can be bounded from above by $k + k^3\log (kD)$.
\end{theorem}

Before we proceed to the proof, let us note that if we write $R(y_0,\ldots,y_m)=\frac{A(y_0,\ldots,y_m)}{B(y_0,\ldots,y_m)}$, where $A,B\in \bbQ[y_0,\ldots,y_m]$, then \cref{thm:simple:ratrec} shows that the following polynomial is cancelling for $\bu^{(1)}$:
\[P(y_0,\ldots,y_m,y_{m+1})=y_{m+1}\cdot B(y_0,\ldots,y_m)-A(y_0,\ldots,y_m).\]
Thus, \cref{thm:simple:ratrec} shows (and in fact, is equivalent to) that every ratrec sequence over $\bbQ(\bx)$ admits a cancelling polynomial that is linear in the last variable (here $y_{m+1}$),
improving upon \cref{thm:ratrec-cancelling}.

% A polynomial $P \in \bbQ[x_0, \dots, x_m]$ is \emph{cancelling} for a sequence $u_n$
% if $P(u_n, \dots, u_{n+m}) = 0$ for all $n \in \bbN$.
% %
% Cancelling polynomials are known to exist for polyrec sequences \cite[Theorem 11]{CadilhacMPPS20},
% and the same proof can be adapted to show that the same holds for ratrec sequences.
% %
% \cref{thm:simple:ratrec} ensures the existence of cancelling polynomials \emph{of degree~1 in~$u^{(1)}_{n+m}$},
% and thus it does not follow from previous results.

% \medskip

The remainder of this section is devoted to the proof of \cref{thm:simple:ratrec} and to a discussion related to it.
In particular, the first part of the theorem (existence) will be proved in \cref{sec:transcendence:A}
and the concrete bound on the depth $m$ will be proved in \cref{sec:transcendence:B}.
%

%\michal{I rephrased the discussion above, hopefully this clarifies the issue with zeroness of the denominators that flagged here before.}

Let us make a few observations about the sequences $F_n^{(1)},\ldots,F_n^{(k)}$. First, a straightforward estimation shows that the degrees of functions $F_n^{(1)},\ldots,F_n^{(k)}$ grow at most single-exponentially in $n$.
% ; the proof is in the appendix.

\begin{restatable}{lemma}{lemRatRecDegreeGrowth}
    \label{lem:ratrec:degree:growth}
    For $n\in \bbN$, let $d_n$ be the maximum degree of $F_n^{(1)}, \dots, F_n^{(k)}$.
    Then $d_n \leq (k \cdot D)^n$.
\end{restatable}

\begin{proof}
    We proceed by induction on $n$. Initially we have $d_0 = 1$ by definition.
    By \cref{def:ratrec} $F_{n+1}^{(i)}$ is obtained by substituting rational functions $F_n^{(1)}, \ldots, F_n^{(k)}$ of degree at most $d_n$
    into a rational function $P_i$ of degree at most $D$.
    Let $P_i = \frac A B$ be the ratio of two polynomials $A, B \in \bbQ[\bx]$
    of degree at most $D$.
    Let $C \in \bbQ[\bx]$ be the least common multiple of all denominators of $F_n^{(1)}, \ldots, F_n^{(k)}$,
    and thus of degree at most $k \cdot d_n$.
    We can then write $F_n^{(1)} = \frac {G^{(1)}} C, \dots, F_n^{(k)} = \frac {G^{(k)}} C$,
    where the numerators $G^{(1)}, \dots, G^{(k)} \in \bbQ[\bx]$ are polynomials of degree also at most $k \cdot d_n$.
%    and similarly let $D \in \bbQ[\bx]$ be the product of all their denominators.
    %
    It follows that both $A(F_n^{(1)}, \ldots, F_n^{(k)})$ and $B(F_n^{(1)}, \ldots, F_n^{(k)})$ can be written as rational functions of the form $\frac {\hat A} {C^D}$, resp., $\frac {\hat B} {C^D}$,
    where the numerators are polynomials $\hat A, \hat B \in \bbQ[\bx]$ of degree at most $D \cdot k \cdot d_n$
    and the same holds for the common denominator $C^D \in \bbQ[\bx]$.
    It follows that $F_{n+1}^{(i)}$ is a rational function of degree $d_{n+1} \leq k \cdot D \cdot d_n$, as required.
\end{proof}

The next lemma is a key property implied by the recurrence: if several consecutive elements of the sequence $F_n^{(i)}$ satisfy some algebraic constraint, then this constraint is also satisfied at every step later in the sequence.

\begin{lemma}[Substitution lemma]
    \label{lem:substitution}
    Suppose $Z(y_0, \dots, y_m) \in \bbQ[y_0, \dots, y_m]$ is a polynomial such that~%
    $Z(F_0^{(i)}(\bx), \dots, F_{m}^{(i)}(\bx)) = 0$.
    Then
    $Z(F_n^{(i)}(\bx), \dots, F_{n+m}^{(i)}(\bx)) = 0$
    for all~$n \in \bbN$.
\end{lemma}

\begin{proof}
    By assumption we have
    \begin{align}
        \label{eq:subst}
        Z(F_0^{(i)}(\bx), \dots, F_{m}^{(i)}(\bx)) = 0.
    \end{align}
    Consider the ring homomorphism $h\colon \bbQ[\bx]\to \bbQ(\bx)$ that maps the variables $x_1,\ldots,x_k$ to rational functions $F_1^{(1)},\ldots,F_1^{(k)}$, respectively. For a rational function $P/Q$ such that $h(Q) \neq 0$, by $h(P/Q)$ we understand the rational function $h(P)/h(Q)$. (Note that such an extension of $h$ to $\bbQ(\bx)$ does not have to be a field homomorphism.) From the definition of the sequence $F_n^{(i)}$ it readily follows that 
    \begin{align*}
     h(F_n^{(i)}) = F_{n+1}^{(i)}, \quad \text{for all } n \in \bbN.
    \end{align*}
    Thus, by applying $h$ to both sides of~\eqref{eq:subst}, we infer that
    \begin{align*}
        Z(F_1^{(i)}(\bx), \dots, F_{m+1}^{(i)}(\bx)) = 0.
    \end{align*}
    We conclude by repeating this reasoning $n$ times.
\end{proof}

%\michal{Rewrote the proof in the language of homomorphisms.}

In the following we introduce some basic terminology about (commutative) fields
(c.f.~\cite[Sec.~II.1]{Nagata:1993}, \cite[Sec.~V.3]{Bourbaki:2003}, or \cite[Sec.~13.1 and 13.2]{DummitFoote:AA:2003} for more details).
Let $\bbE, \bbF$ be two fields.
When $\bbE \subseteq \bbF$ we say that $\bbF$ is a \emph{field extension} of $\bbE$,
which is called the \emph{base field}.
The \emph{degree} of $\bbF$ over $\bbE$, written $\deg_\bbE \bbF$, is the dimension of $\bbF$ as a vector space over the base field $\bbE$.
For instance, $\bbQ(\sqrt 2)$ has degree 2 over $\bbQ$ (its elements can be put in the form $a + b \cdot \sqrt 2$)
and $\bbQ(\sqrt[3]2)$ has degree 3 (its elements can be put in the form $a + b \cdot \sqrt[3] 2 + c \cdot (\sqrt[3] 2)^2$).
Field extensions need not have finite degree.
For instance, $\bbQ(\pi)$ and $\bbQ(x)$ are two field extensions of $\bbQ$ of infinite degree.
The degree is multiplicative:

\begin{lemma}[\protect{c.f.~\cite[Theorem 14]{DummitFoote:AA:2003}}]
    \label{lem:degree:multiplicative}
    Consider field extensions $\bbE \subseteq \bbF \subseteq \bbG$.
    Then, $\deg_\bbE \bbG = \deg_\bbE \bbF \cdot \deg_\bbF \bbG$
    (even for infinite degrees).
\end{lemma}

An element $f \in \bbF$ is \emph{algebraic} over the base field $\bbE$ if there is a nonzero polynomial $P(x) \in \bbE[x]$ s.t.~$P(f) = 0$.
The field extension $\bbF$ is \emph{algebraic} over the base field $\bbE$ if every element in $\bbF$ is algebraic over $\bbE$.
%

% We will use the following well-known fact.

% \begin{lemma}
%     Every field extension of finite degree is algebraic.
% \end{lemma}

% % \marysia{Would you like me to find a reference for all this basic stuff? (Probably Lang's Algebra will do.)}
% % \lorenzo{We are now already using three different books, perhaps Nagata or Bourbaki or Dummit-Foote would do?}

% The converse does not hold since there are algebraic extensions of infinite degree. For instance, the set of all algebraic rational numbers $\bar \bbQ$ is algebraic but of infinite degree over $\bbQ$.
% However, a field extension has a finite degree if and only if it is generated by finitely many algebraic elements.
% \filip{algebraic real numbers? Also we never use Lemma~\ref{lem:deg:multiplicative} but the remark above, why not state it as a lemma?}

% \marysia{Algebraic element over the considered base field, not necessarily algebraic reals. And I also cannot see where this is used explicitly, so maybe we don't have to call it a lemma.}

Let $\bbF$ be a field extension of $\bbE$.
A subset $\set{f_1, \dots, f_n} \subseteq \bbF$ of elements of $\bbF$ is \emph{algebraically independent} over $\bbE$
if there is no nonzero polynomial $P(x_1, \dots, x_n) \in \bbE[x_1, \dots, x_n]$ such that~$P(f_1, \dots, f_n) = 0$.
The \emph{transcendence degree} of $\bbF$ over $\bbE$, denoted $\trdeg_{\bbE} \bbF$,
is the largest number of elements of $\bbF$ which are algebraically independent over $\bbE$.
Note that $\bbF$ is algebraic over $\bbE$ if and only if $\trdeg_{\bbE} \bbF = 0$.
Like the algebraic degree is multiplicative, the transcendence degree is additive:
\begin{lemma}[\protect{c.f.~\cite[Corollary to Theorem 4, A.5.111]{Bourbaki:2003}}]
    \label{lem:degree:additive}
    Consider field extensions $\bbE \subseteq \bbF \subseteq \bbG$.
    Then, $\trdeg_\bbE \bbG = \trdeg_\bbE \bbF + \trdeg_\bbF \bbG$.
\end{lemma}

In the following we will always take as the base field $\bbE = \bbQ$,
in which case we will write just $\trdeg \bbF$ instead of $\trdeg_\bbQ \bbF$.
For example, $\trdeg \bbQ(\sqrt 2) = 0$ because $\sqrt 2$ is an algebraic number over $\bbQ$,
$\trdeg \bbQ(\sqrt 2, \pi) = 1$ because $\pi$ is a transcendental number,
and $\trdeg \bbQ(x_1, \dots, x_n) = n$.

% \begin{lemma}
%     Consider field extensions $\bbE \subseteq \bbF \subseteq \bbG$.
%     Then, $\trdeg_\bbE \bbF = \trdeg_\bbE \bbG$ if, and only if, $\bbG$ has finite degree over $\bbF$.
% \end{lemma}

%\marysia{Technically that's not true. You may have many 1-dim spaces, then jump to dim 2, then many 2-dim spaces, etc.}

Given a field extension $\bbF$ over $\bbE$ and elements $f_1, \dots, f_n \in \bbF$,
let $\bbE(f_1, \dots, f_n)$ be the smallest field extension over $\bbE$ containing $f_1, \dots, f_n$.
If $\bbF = \bbE(f_1, \dots, f_n)$, then we say that $\bbF$ is \emph{finitely generated over $\bbE$}
(with \emph{generators} $f_1, \dots, f_n$).

The motivation to look at field extensions is that a ratrec system naturally defines the following sequence of field extensions
\begin{align}\label{eq:chain}
    \bbQ \subseteq \bbF_0 \subseteq \bbF_1 \subseteq \ldots \subseteq \bbQ(\bx),
\end{align}
where $\bbF_0=\bbQ(x_1)$ and $\bbF_{n+1} = \bbF_n(F_{n+1}^{(1)}(\bx))$
%$\bbQ(F_0^{(1)}(x_1,\ldots,x_k),\ldots, F_{n+1}^{(1)}(x_1,\ldots,x_k))$
for $n \in \bbN$.

\subsection{Ascending sequences of field extensions}
\label{sec:transcendence:A}

In this section we prove the following Noether-like result.

\begin{theorem}
    \label{thm:noether}
    Consider any ascending sequence of field extensions of the form
    \begin{align*}
        \bbQ \subseteq \bbF_0 \subseteq \bbF_1 \subseteq \ldots \subseteq \bbQ(x_1,\ldots, x_k).
    \end{align*}
    Then the sequence eventually stabilises: there exists $n_0$ such that~$\bbF_{n_0} = \bbF_{n_0+1} =\bbF_{n_0+2} = \ldots$.
\end{theorem}

The crucial reason for the result above is that the number $k$ of variables is fixed.
In the proof of \cref{thm:noether} we use the following result on finitely generated extensions.

\begin{lemma}[\protect{c.f.~\cite[A.5.118, Cor.~3]{Bourbaki:2003}}]
    \label{lem:bourbaki}
    If $\bbG$ is a finitely generated extension over $\bbE$,
    then every subextension $\bbE \subseteq \bbF \subseteq \bbG$ of $\bbG$ over $\bbE$ is also finitely generated.
\end{lemma}

\begin{proof}[Proof of \cref{thm:noether}]
    First of all, observe that
    \begin{align*}
        \trdeg \bbF_n \leq \trdeg \bbQ(x_1,\ldots,x_k)= k,\quad\textrm{for all }n.
    \end{align*}
    Hence, there is $n_1$ such that~$\trdeg \bbF_{n_1} = \trdeg \bbF_{n_1 + 1} = \cdots$.
    Let $\bbF_\infty \coloneqq \bigcup_{n = 0}^\infty \bbF_n$ and consider the ascending sequence
    \begin{align}
        \label{eq:extensions}
        \bbF_{n_1} \subseteq \bbF_{n_1+1} \subseteq \cdots \subseteq \bbF_\infty.
    \end{align}
    %
    %Note that $\trdeg \bbF_{\infty}=\trdeg \bbF_{n_1}$, hence all field extensions in \eqref{eq:extensions} are algebraic over $\bbF_{n_1}$.
    %
    We have $\trdeg \bbF_{n_1} = \trdeg \bbF_{n_1+i}$ for all $i>0$,
    and, by \cref{lem:degree:additive},
    $\trdeg_{\bbF_{n_1}} \bbF_{n_1 + i} = 0$, i.e., $\bbF_{n_1+i}$ is algebraic over $\bbF_{n_1}$.
    Moreover, $\bbF_{\infty}$ is also algebraic over $\bbF_{n_1}$ because any element of $\bbF_{\infty}$ belongs to some $\bbF_{n_1+i}$.
    %\filip{I think here we use that if $\bbF$ extends $\bbE$ and $\trdeg \bbF = \trdeg \bbE$ then $\trdeg_{\bbE} \bbF = 0$ (it's trivial but it took me a while)}
    %\marysia{Yes, we use this fact. I rephrased the sentence above to make it explicit. We may also add a one-sentence argument.}
    Since $\bbF_{n_1} \subseteq \bbQ(x_1,\ldots,x_k)$ is a finitely generated extension of $\bbF_{n_1}$
    and $\bbF_{n_1} \subseteq \bbF_\infty \subseteq \bbQ(x_1,\ldots,x_k)$ is a subextension of $\bbF_{n_1}$,
    by \cref{lem:bourbaki} we have that $\bbF_\infty$ is also a finitely generated extension of $\bbF_{n_1}$.
    In other words, there are generators $f_1, \dots, f_m \in \bbF_\infty$
    such that
    \begin{align*}
\bbF_\infty = \bbF_{n_1}(f_1, \dots, f_m).
              \end{align*}
    Since the generators $f_1,\ldots,f_m$ are algebraic over $\bbF_{n_1}$,
    $\bbF_\infty$ is an algebraic extension of finite degree over $\bbF_{n_1}$ by \cref{lem:degree:multiplicative}.
    (Concretely, an upper bound for the degree is the product of the degrees of minimal polynomials of the generators $f_1, \dots, f_n$.)
    It follows that the sequence in \eqref{eq:extensions} is an ascending sequence of vector subspaces of $\bbF_{\infty}$, where we treat $\bbF_{\infty}$ as a vector space over~$\bbF_{n_1}$. Since the dimension of $\bbF_{\infty}$ as a vector space over $\bbF_{n_1}$ is finite, this sequence must eventually stabilize at $\bbF_{n_0}$ for some $n_0\geq n_1$.
 %   \filip{Why cannot we write that the last part just follows from Lemma~\ref{lem:degree:multiplicative}?}
 %   \marysia{We can if you think that the last two sentences are not necessary. I added the reference to \cref{lem:degree:multiplicative} where we actually use it.}
\end{proof}

%\michal{I removed the Noether result for vector spaces and inlined it into the proof of \cref{thm:noether}.}

We now prove the existence part of \cref{thm:simple:ratrec} using \cref{thm:noether}.

\begin{proof}[\protect{Proof (of the first part of \cref{thm:simple:ratrec})}]
    By \cref{thm:noether}, the sequence in \eqref{eq:chain} stabilizes at some $\bbF_m$,
    that is, \[\bbF_m = \bbF_{m+1} = \bbF_{m}(F_{m+1}^{(1)}(\bx)).\]
    Therefore, we have $F_{m+1}^{(1)}(\bx) \in \bbF_{m}$. Noting that $\bbF_m=\bbQ(F_0^{(1)}(\bx),\ldots,F_m^{(1)}(\bx))$, we see that $F_{m+1}^{(1)}(\bx)$ can be expressed as a rational function of the generators:
    There exists a rational function $R \in \bbQ(y_0, \dots, y_m)$ such that
    \[F_{m+1}^{(1)}(\bx) = R(F_0^{(1)}(\bx), \dots, F_m^{(1)}(\bx)).\]
    We may now apply \cref{lem:substitution} to the numerator of the rational function $R(y_0,\ldots,y_m)-y_{m+1}$, thus obtaining that
    \begin{align*}
        F_{n + m + 1}^{(1)}(\bx) = R(F_n^{(1)}(\bx), \dots, F_{n + m}^{(1)}(\bx)),\quad \textrm{for every }n\in \bbN.
    \end{align*}
% \filip{commented last line about sequences}
%     %
%     Since $u^{(1)}_m = F_m^{(1)}(u^{(1)}_0, \dots, u^{(k)}_0)$,
%     this implies that $u^{(1)}_{n+m+1} = R(u^{(1)}_n, \dots, u^{(1)}_{n+m})$ for every $n \in \bbN$,
%     as required.
\end{proof}

%    \lorenzo{TODO: argue what happens if a division by zero occurs when replacing numbers for variables.}

\subsection{Upper bound}
\label{sec:transcendence:B}

We now move to the second, quantitative part of the proof of \cref{thm:simple:ratrec}: we need to prove that $m$ is bounded from above by $k+k^3 \log (kD)$. For this, we inspect the proof of \cref{thm:noether} in the special case of the chain of extensions~\eqref{eq:chain} given by a ratrec system. The first observation is that the sequence of transcendence degrees stabilises very quickly.

\begin{lemma}
    \label{lem:trdeg}
    The transcendence degrees $\trdeg \bbF_n$ of the sequence \eqref{eq:chain} stabilise after at most $k$ steps.
\end{lemma}
\begin{proof}
    As argued, $\trdeg \bbF_n \leq \trdeg \bbQ(x_1,\ldots,x_k)=k$ for all $n$.
    The next extension $\bbF_{n+1}$ is obtained by adding a new rational function $F_{n+1}^{(1)}(x_1, \dots, x_k)$ to the previous extension $\bbF_n$.
    This immedately shows that $\trdeg \bbF_n \leq \trdeg \bbF_{n+1} \leq \trdeg \bbF_n + 1$.
    We argue that if $\trdeg \bbF_{m+1} = \trdeg \bbF_m$ for some $m$,
    then the transcendence degree cannot change anymore: $\trdeg \bbF_m=\trdeg \bbF_{m+1}=\trdeg \bbF_{m+2}=\cdots$. Note that this will conclude the proof, because then the transcendence degree can increase at most $k$ times before eventually stabilizing.
    
    Since $\trdeg \bbF_{m+1} = \trdeg \bbF_m$,
    it follows that $F_{m+1}^{(1)}(\bx)$ is algebraic over $\bbF_m$,
    which means that it satisfies $P(F_{m+1}^{(1)}(\bx)) = 0$ for some nonzero polynomial $P(x) \in \bbF_n[x]$.
    By clearing out denominators, there is a nonzero polynomial $Z(y_0, \dots, y_{m+1}) \in \bbQ[y_0, \dots, y_{m+1}]$ such that
    \begin{align}
        \label{eq:expr}
        Z(F_0^{(1)}(\bx), \dots, F_{m+1}^{(1)}(\bx)) = 0.
    \end{align}
    % %
    By \cref{lem:substitution} we have 
    \begin{align*}
Z(F_n^{(1)}(\bx), \dots, F_{n+m+1}^{(1)}(\bx)) = 0\quad\textrm{for every }n \in \bbN.     
    \end{align*}
    This means that $F_{n+m+1}^{(1)}(\bx)$ is algebraic over $\bbF_{n+m}$, implying
    $$\trdeg \bbF_{n+m+1}=\trdeg \bbF_{n+m}(F^{(1)}_{n+m+1}(\bx)) =\trdeg \bbF_{n+m}.$$
    This concludes the proof.
\end{proof}

%\michal{I shortened the last part of the proof, please check.}
%\marysia{All good.}

%     and thus by \cref{lem:deg:multiplicative} also over $\bbF_n$.
%     We conclude that $\bbF_{n+1} \subseteq \bbF_{n+2} \cdots$
%     are all algebraic extensions over $\bbF_n$,
%     which is to say, $\trdeg_{\bbF_n} \bbF_{n+1} = \trdeg_{\bbF_n} \bbF_{n+2} = \cdots = 0$,
%     and thus $\trdeg_\bbQ \bbF_n = \trdeg_\bbQ \bbF_{n+1} = \cdots$, as required.

Note that even when the transcendence degrees of the fields in~\eqref{eq:chain} stabilise, it may still take several  further steps until the fields themselves eventually stabilise. We will later give an example that this may indeed happen. 

We are left with estimating the degrees of field extensions after the transcendence degree in the chain~\eqref{eq:chain} stabilises. For this, we use the following two results.

\begin{lemma}[\protect{c.f.~\cite[Lemma 3.4]{Muller-Quade:Steinwandt:JSC:1999}}]
    \label{lem:Muller-Quade:Steinwandt}
    Let $f \in \bbQ(x_1, \dots, x_k)$ be algebraic over $$\bbF_n = \bbQ(F_0(x_1, \dots, x_k), \dots, F_n(x_1, \dots, x_k)).$$
    Then there is a polynomial $Z(x) \in \bbF_n[x]$ of degree at most $\deg F_0 \cdots \deg F_n$ s.t.~$Z(f) = 0$.
\end{lemma}

\begin{lemma}[\protect{c.f.~\cite[Exercise III.A.2]{Nagata:1993}}]
    \label{lem:Nagata}
    Let $\bbQ \subseteq \bbF \subseteq \bbQ(x_1, \dots, x_k)$ be a subextension of $\bbQ(x_1, \dots, x_k)$ over $\bbQ$
    of transcendence degree $r \coloneqq \trdeg_\bbQ \bbF$.
    Then there are (algebraically independent) rational functions $f_1, \dots, f_r \in \bbQ(x_1, \dots, x_k)$ such that~$\bbF \subseteq \bbQ(f_1, \dots, f_r)$.
\end{lemma}

\begin{lemma}\label{lem:stabilization}
    The sequence \eqref{eq:chain} eventually stabilizes after at most $k + k^3 \log (kD)$ steps.
\end{lemma}
\begin{proof}
    By \cref{lem:trdeg}, there exists $j_1\leq k$ such that~$$r \coloneqq \trdeg \bbF_{j_1} = \trdeg \bbF_{j} \quad \textrm{for all }j \geq j_1.$$
    In particular, all field extensions $\bbF_j$ for $j \geq j_1$ are algebraic over $\bbF_{j_1}$.
%    Consider the field extension $\bbF_\infty$ over $\bbF_{j_1}$ defined as
    %
%    \begin{align}
%        \bbF_{j_1} \subseteq \bbF_\infty \coloneqq \bigcup\limits_{j=0}^{\infty} \bbF_j \subseteq \bbQ(x_1, \dots, x_k).
%    \end{align}
    %
%    Since every element of~$\bbF_{\infty}$ is algebraic over~$\bbF_{j_1}$ as an element of some $\bbF_j$, this extension is also algebraic over $\bbF_{j_1}$.
%    In particular, $\trdeg \bbF_\infty = r$.
    As in the proof of \cref{thm:noether}, consider the field extension $\bbF_\infty =  \bigcup\limits_{j=0}^{\infty} \bbF_j$ over $\bbF_{j_1}$, which is algebraic. In particular, $\trdeg \bbF_\infty = r$.
%    \filip{Maybe to avoid repeating the argument we could add a Lemma for the first part of the proof of \cref{thm:noether}}
%    \marysia{It seems here simpler to repeat this than extract it as a lemma. I tried to shorten it, referring to the proof of \cref{thm:noether}. If it doesn't work well, I'm ok with the previous version (commented out) or making it a lemma.}
    %
    By \cref{lem:Nagata}, there are rational functions $f_1, \dots, f_r \in \bbQ(x_1, \dots, x_k)$ s.t.~$\bbF_\infty \subseteq \bbQ(f_1, \dots, f_r)$.
    Since $\bbQ(f_1, \dots, f_r)$ has the same transcendence degree $\trdeg \bbQ(f_1, \dots, f_r) = r$ as $\bbF_{j_1}$,
    it follows that the $f_i$'s are algebraic over $\bbF_{j_1}$.
    %\lorenzo{Why the $f_i$'s are algebraic over $\bbF_{j_1}$?}
    %\marysia{Because $\trdeg \bbQ(f_1,\ldots,f_r)$ is $r$, which is the same as $\trdeg F_{j_1}$ -- this means that the extension is algebraic.}
    By \cref{lem:Muller-Quade:Steinwandt}, each $f_i$ is algebraic of degree at most $\deg F_0 \cdots \deg F_{j_1}$ over $\bbF_{j_1}$.
    It follows that $\bbQ(f_1, \dots, f_r)$ is an algebraic extension of degree at most $d = (\deg F_0 \cdots \deg F_{j_1})^r$ over $\bbF_{j_1}$.
    Thus the chain
    \begin{align*}
        \bbF_{j_1} \subseteq \bbF_{j_1+1} \subseteq \bbF_{j_1+2} \subseteq \cdots
    \end{align*}
    is such that the degree of any $\bbF_{j_1+t}$ over $\bbF_{j_1}$ is at most $d$. In particular, all extensions in this chain are algebraic. We show that it stabilizes after at most~$d$ steps. Assume that for some $t \geq 0 $ we have $\bbF_{j_1 + t} = \bbF_{j_1 + t + 1}$, that is $F^{(1)}_{j_1 + t + 1}(\bx) \in \bbF_{j_1 + t}$. Thus, there is a rational function $R \in \bbF_{j_1}(y_1,\ldots,y_t)$ such that $F^{(1)}_{j_1 + t + 1}(\bx) = R(F^{(1)}_{j_1 + 1}(\bx),\ldots, F^{(1)}_{j_1 + t}(\bx))$. Then by applying \cref{lem:substitution} to the numerator of $R(y_1,\ldots,y_t) - y_{t+1}$, we may express $F^{(1)}_{j_1 + n + t + 1}(\bx)$ as a rational function of $F^{(1)}_{j_1 + n + 1}(\bx),\ldots, F^{(1)}_{j_1 + n + t}(\bx)$, i.e., elements of $\bbF_{j_1+n+t}$. Hence an equality in the field chain implies stabilization at this point. 
    
    Since the degree grows at each step before stabilization and the degree is multiplicative, the chain stabilizes after at most $\log d$ steps, at $\bbF_{j_0}$ for some $j_0\leq j_1+\log((\deg F_0 \cdots \deg F_{j_1})^r)$.
%     \marysia{Actually, we can have $\log d$ steps here, because the degree is multiplicative.}
    By \cref{lem:ratrec:degree:growth} and since $j_1\leq k$ and $r\leq k$, we have, as required,
    \begin{align*}
        j_0
            & \leq j_1+r \log(\deg F_0 \cdots \deg F_{j_1}) \\
            & \leq k+ k\log((kD)^0\cdots (kD)^k) \\
            & \leq k + k^3\log (kD). \qedhere
    \end{align*}
\end{proof}

We are ready to provide the proof of the quantitative bound promised in \cref{thm:simple:ratrec}.

\begin{proof}[\protect{Proof (of the second part of \cref{thm:simple:ratrec})}]
    It suffices to observe that $m$ in the proof the first part of \cref{thm:simple:ratrec}
    can be bounded by $k + k^3\log (kD)$ thanks to \cref{lem:stabilization}.    
\end{proof}

% \marysia{I don't know if this example is still useful, we may delete it or move to some better place.}
% \lorenzo{It think it's very nice and we need it to make the point that there can be nontrivial algebraic extensions.}
% \michal{I think it fits the story nicely. I reshuffled the story a bit here.}

%  
% The example below shows that the ascending sequence from \eqref{eq:chain} does not always stabilize just when $\trdeg  \bbF_j$ reaches its maximal value, but there can be a subsequence of non-trivial algebraic extensions.

We finish this section by giving an example that shows that in the proof of \cref{lem:stabilization}, it may happen that $j_1<j_0$, that is, after the stabilisation of the transcendence degree, there can be several non-trivial algebraic extensions until the fields themselves stabilise. Consider the poly-rec system
    \begin{align*}
        \left\{\begin{array}{lcl} 
            u^{(1)}_{n+1}   &=& (u_n^{(1)})^2 + (u_n^{(2)})^2, \\
            u^{(2)}_{n+1}   &=& u_n^{(1)} + u_n^{(2)}.
        \end{array}\right.
    \end{align*}
We have $F^{(1)}_0 = x_1, F^{(2)}_0 = x_2$, then $F^{(1)}_1 = x_1^2 + x_2^2, F^{(2)}_1 = x_1+x_2$ and $F^{(1)}_2 = (x_1^2+x_2^2)^2 + (x_1+x_2)^2$. The chain~\eqref{eq:chain} starts with
\[\bbQ \subseteq \bbF_0 =\bbQ(x_1) \subseteq \bbF_1 = \bbF_0(x_1^2+x_2^2) = \bbQ(x_1,x_2^2) \subseteq \bbF_2.\]
Note that $\trdeg \bbF_0 = 1$ and $\trdeg  \bbF_1= 2$, which is the maximum value. However, the next extension $\bbF_1 \subseteq \bbF_2 = \bbF_1(F^{(1)}_2) = \bbF_1(x_1x_2)$ is non-trivial, because $x_1x_2$ does not belong to $\bbQ(x_1,x_2^2)$. In fact, it is algebraic of degree~2.
% 
% As an application of \cref{lem:stabilization} we prove the effective bound in \cref{thm:simple:ratrec}.

\input{counterexample}

%% file: counterexample.tex
\subsection{Obstacles towards the zeroness problem for polyrec sequences}
\label{subsec:obstacles}

\cref{thm:simple:ratrec} suggests the following algorithm for deciding zeroness of a polyrec sequence $\bu$. Suppose the dimension of $\bu$ is $k$ and the degree is $D$. We compute the first $p+2$ entries of~$\bu$, where $p=k+ k^3\lceil\log (kD)\rceil$, and we verify whether all of them are zero. Obviously, if one of them is non-zero, then $\bu$ is non-zero. Otherwise, by \cref{thm:simple:ratrec}, we expect that there is a rational function $R(y_0,\ldots,y_m)$ for some $m\leq p$ such that
\begin{equation}\label{eq:incorrect}
        u_{n+m+1} = R(u_n,u_{n+1},\ldots,u_{n+m})\qquad\textrm{for all }n\in \bbN.
    \end{equation}
In particular,
\begin{equation*}
0=u_{m+1}=R(u_0,\ldots,u_m)=R(0,\ldots,0). 
\end{equation*}
Consequently,
\begin{equation*}
u_{m+2}=R(u_1,\ldots,u_{m+1})=R(0,\ldots,0)=0,
\end{equation*}
and a straightforward induction shows that $u_n=0$ for all $n\in \bbN$. So we can declare that $\bu$ is the zero sequence.

The reasoning above is incorrect for the following reason.  By \cref{thm:simple:ratrec}, there is a rational function $R\in \bbQ(y_0,\ldots,y_m)$ such that~\eqref{eq:incorrect} holds when both sides are treated symbolically, as rational functions over a set of $k$ variables $\bx$ that denote the vector of initial entries of the polyrec system defining $\bu$. However, $R$ is a rational function, hence when the variables are substituted with actual entries of the sequence $\bu$, we may get an accidental $0$ in the denominator of the right hand side. In other words, assertion~\eqref{eq:incorrect} may be incorrect due to the right hand side being ill-defined, which renders the remainder of the reasoning flawed. To exemplify the problem we now present an example where this situation actually occurs.

Fix some $d\in \bbN$, and let
\[P(x)=x(x-1)\dots (x-d+1).\]
Define the sequence $\bu$ by setting
\[u_n=P(n)\qquad\text{for all }n\in \bbN.\]
It is straightforward to see that $\bu$ is polyrec of dimension $2$ and degree $d$: one can simply use one auxiliary sequence $\bv$ with $v_n=n$.

Observe that if instead of setting $v_0=0$, we set $v_0=x$ for a formal variable $x$, the same polyrec system defines a sequence of polynomials $\widehat{\bu}$ over $x$ defined as
\[\widehat{u}_n=P(x+n)\qquad\text{for all }n\in \bbN.\]
(Here, we also set initial condition $\widehat{u}_0=P(x)$.)
Now, we may apply the reasoning behind \cref{thm:simple:ratrec} to find the rational function $R(y_0,y_1)\in \bbQ(y_0,y_1)$, defined as
\[R(y_0,y_1)=y_1\cdot \frac{(d+1)\cdot y_1-y_0}{y_1+(d-1)\cdot y_0},\]
such that
\[\widehat{u}_{n+2}=R(\widehat{u}_{n},\widehat{u}_{n+1})\qquad\textrm{for all }n\in \bbN.\]
This, however, should be regarded as an equality of two rational functions over the variable~$x$, which means that we cannot infer that
\[u_{n+2}=R(u_n,u_{n+1})\qquad\textrm{for all }n\in \bbN,\]
because the right hand side can be undefined for specific values; and indeed, $R(0,0)$ is undefined. The flawed reasoning from the beginning of this section would suggest that in order to verify the zeroness of $\bu$, it suffices to check that the first three entries of $\bu$ are zero. However, we have $u_0=u_1=\dots=u_{d-1}=0$ and $u_d=d!\neq 0$, so the algorithm would provide an incorrect answer.

Notice that if we had a promise that we never encounter a division by zero
when recursively applying \eqref{eq:incorrect} from the given initial conditions,
then the na\"ive zeroness algorithm presented at the beginning of the section would be sound.
(The na\"ive algorithm is complete even without the promise.)
However, deciding whether no division by zero occurs is essentially the Skolem problem for polyrec sequences,
which, as mentioned in the introduction, is a long-standing open problem.

We are hopeful that the problem with accidentally hitting a singularity of $R$ when starting from a polyrec sequence,
as present in the example above,
can somehow be circumvented, hence we state the following conjecture.

\begin{conjecture}\label{conj:elementary}
 There is an elementary function $g\colon \bbN\to \bbN$ such that the following holds. Suppose $\bu$ is a polyrec sequence of dimension at most $N$ and degree at most $N$ such that $u_n=0$ for all $n\leq g(N)$. Then $u_n=0$ for all $n\in \bbN$. 
\end{conjecture}

Note that a positive resolution to \cref{conj:elementary} would immediately imply that the complexity of the zeroness problem for polyrec sequences is elementary.

%% file: conclusions.tex
We believe that ratrec is a natural class of sequences with various promising questions deserving further investigation. Questions about decision problems are more natural for polyrec sequences due to their connection to polynomial automata and the issues with division by $0$ in ratrec discussed in the introduction. Nevertheless, as discussed in this paper, understanding the properties of ratrec might lead to concrete complexity results for polyrec. The most natural problem for future work is to overcome the obstacles discussed in \cref{subsec:obstacles}.

%% file: lower.tex
In order to speak about computational aspects of poly-rec sequences, we need to fix how they are encoded on input. For robustness, we choose to use arithmetic circuits. Formally, for a fixed field $\bbF$, a polynomial $P\in \bbF[x_1,\ldots,x_k]$ is encoded by a circuit $C$ that may use the following gates:
\begin{itemize}
 \item binary addition and multiplication gates;
 \item nullary input gates, bijectively labelled with variables $x_1,\ldots,x_k$; and
\item nullary constant gates, each labelled with an element of~$\bbF$.
\end{itemize}
Note that subtraction can be emulated using addition and multiplication by the constant $-1$.
One of the gates is designated as the output gate. 
Given a valuation of variables with elements of $\bbF$, the values of the gates can be computed as expected, and the value yielded by the circuit $C$ is the one computed for the output gate.

In this section we prove the following lower bound.

\begin{theorem}\label{theorem:pspace}
        For every fixed field $\bbF$, the zeroness problem for polyrec sequences over $\bbF$ is \PSPACE-hard.
\end{theorem}

The lower bound claimed in the introduction follows from the theorem above by taking $\bbF = \bbQ$.
Note also that together with \cref{thm:finite-field} below, we can conclude that the problem is actually \PSPACE-complete for every fixed {\em{finite}} field $\bbF$.

\begin{theorem}\label{thm:finite-field}
 For every fixed finite field $\bbF$, the zeroness problem for polyrec sequences over $\bbF$ is in \PSPACE.
\end{theorem}
\begin{proof}
    Let $m$ be the cardinality of $\bbF$; note that $m$ is a fixed constant.
    A standard periodicity argument, \eg as in the proof of~\cite[Theorem 4.1]{cadilhac2021polynomial}, shows that if $\bu$ is a polyrec sequence of dimension $k$,
    then it is zero if, and only if, it is zero for the first $m^k$ steps.
    We can check the latter condition by storing in memory a $k$-tuple of values and computing the first $m^k$ values of the sequence,
    which takes an amount of space which is polynomial in $k$.
\end{proof}

\subsection{Extended polyrec sequences}

In the reductions leading to the lower bound of \cref{theorem:pspace} it is convenient to construct polyrec sequences according to a definition slightly more general than what we allowed in \cref{def:polyrec}.
Namely, the definition of an \emph{extended polyrec system} is the same as before,
except that we generalize the format of the $i$th equation $u^{(i)}_{n+1} = P_i(\cdots)$ by allowing $u^{(i)}_{n+1}$ to additionally depend on $u^{(1)}_{n+1}, \dots, u^{(i-1)}_{n+1}$. Thus, the $i$th equation takes the form:
        \begin{equation}\label{eq:ext-polyrec}
                u^{(i)}_{n+1} = P_i(u^{(1)}_n, \dots, u^{(k)}_n, u^{(1)}_{n+1}, \dots, u^{(i-1)}_{n+1}),
        \end{equation}
where now $P_i$ is a polynomial in $k+i-1$ variables. This more relaxed definition will help focus on the important aspects of the reduction presented in the rest of this section. The following lemma shows that the modification does not affect the complexity of the zeroness problem.

\begin{lemma}\label{lem:ext-elimination}
 Suppose $\bu$ is a sequence defined by an extended polyrec system $S$ of dimension $k$, where each polynomial $P_i$ is represented by circuit $C_i$. Then given the circuits $C_i$, one can in polynomial time construct a circuit $C$ that represents a polyrec system $S'$ of dimension $k$ that also defines $\bu$ (with the same initial condition as $S$).
\end{lemma}
\begin{proof}
 Let the input gates of circuit $C_i$ be labelled with $x_1,\ldots,x_k,z_1,\ldots,z_{i-1}$, where variables $z_1,\ldots,z_{i-1}$ respectively correspond to the values $u^{(1)}_{n+1}, \dots, u^{(i-1)}_{n+1}$ in~\eqref{eq:ext-polyrec}. Construct the circuit $C$ from the union of circuits $C_1,\ldots,C_k$ by performing the following operations for each $i\in \{1,\ldots,k\}$:
 \begin{itemize}
  \item Fuse all input gates labelled $x_i$ in circuits $C_1,\ldots,C_k$ into a single input gate labelled $x_i$.
  \item Fuse the output gate of $C_i$ with all input gates labelled $z_i$ in circuits $C_{i+1},\ldots,C_k$.  
 \end{itemize}
 The output gates of $C$ are the output gates of $C_1,\ldots,C_k$. (Formally, we assumed that output gates must have fan-out $0$, but this can be easily obtained by making a copy of each output gate.) It is straightforward to verify that the polyrec system $S'$ that $C$ represents defines the same $k$-tuple of sequences as $S$ under the same initial condition.
\end{proof}

\subsection{Reduction}

We now proceed to the proof of \cref{theorem:pspace}. Let us fix the field $\bbF$; in the reduction we will use only two constants from $\bbF$, namely $0$ and $1$.
%
% In particular the two element field is $\bbF_2 = \set{0,1}$. 
% The construction will be for sequences over the field $\bbF_2$. The following claim shows that it can be adapted to any other field.
% 
% \begin{claim}
% Fix a polynomial $P \in \bbF_2[x_1,\ldots,x_k]$. Consider any field $\bbF$. Then there is a polynomial $P' \in \bbF[x_1,\ldots,x_k]$ such that $P(a_1,\ldots,a_k) = P'(a_1,\ldots,a_k)$ for every $(a_1,\ldots, a_k) \in \set{0,1}^k$. Moreover, the size of the circuit encoding of $P'$ is polynomial in the size of the circuit encoding of $P$.
% \end{claim}
% 
% \begin{proof}
% It suffices to change the adding gates $z = x + y$ in $P$ into $z = (1-x)y + x(1-y)$.
% \end{proof}
%
We reduce from the validity problem for Quantified Boolean Formulas (\QBF),
which is known to be \PSPACE-complete (see, e.g., \cite[Theorem 19.1]{Papadimitriou:CC:1994}).
Recall that the \QBF validity problem amounts to determine whether a given \QBF of the form
\begin{align}\label{eq:qbf}
        \psi = \exists x_{1} \ \forall x_2\ \ldots\ Q_k x_k \ \varphi(x_1,\ldots,x_k)
\end{align}
is true, where $\varphi(x_1,\ldots,x_k)$ is quantifier-free,
the variables with odd indices are quantified existentially,
the remaining variables are quantified universally,
and $Q_k$ is either $\exists$ or $\forall$ depending on the parity of $k$.
Hence, we are given a \QBF $\psi$ and we wish to construct, in polynomial time, a polyrec system $S$ and its initial condition that define a sequence $\bu$ over $\bbF$ such that the zeroness of $\bu$ is equivalent to the invalidity of $\psi$. By \cref{lem:ext-elimination}, it suffices to construct an extended polyrec system $S$ with this property, where each polynomial $P_i$ involved is represented by a separate circuit.
In the following, the {\em{size}} of an extended polyrec system is the total size of its representation through circuits, which is constructed implicitly.

In the reduction it will be convenient to consider formulas obtained by fixing the truth values of a subset of the bound variables of $\psi$.
For every $i \in \set{0,\ldots,k}$ and $c_{i+1},\ldots,c_{k} \in \set{0,1}$ we define the formula
\begin{align*}%\label{eq:formula}
        \psi|_{c_{i+1},\ldots,c_{k}} = \exists x_1 \forall x_2 \cdots Q_{i} x_i\; \varphi(x_{1},\ldots,x_{i},c_{i+1},\ldots,c_{k}),
\end{align*}
where $Q_i$ is either $\exists$ or $\forall$ depending on the parity of $i$.
In particular, for $i = k$ we get back $\psi$,
and for $i=0$ the formula $\psi|_{c_1,\ldots,c_k}$ reduces to the truth value of $\varphi(c_1,\ldots,c_k)$.
We encode a quantifier Boolean formula $\varphi$ into a polynomial $P_\varphi$
using the following simulation of Boolean operators $\neg$, $\wedge$ and $\vee$ by arithmetic operations:
%. Given two polynomials $P$ and $Q$:
\begin{align}
        \label{eq:boolean}
        \begin{split}
                P_x &= x, \\
                P_{\neg \varphi} &= 1 - P_\varphi, \\
                P_{\varphi \land \psi} &= P_\varphi \cdot P_\psi, \\
                P_{\varphi \lor \psi} &=  P_{\neg (\neg \varphi \land \neg \psi)}.
        \end{split}
\end{align}
For example, $\tau(x,y) = (x\wedge y) \vee \neg y$ is encoded as $P_\tau(x,y) = 1 - (1-xy)y$.
The following straightforward claim shows that with the standard interpretation of $1$ and $0$ representing true, resp., false,
such polynomials evaluate as expected. Note that this claims holds in any fixed field.

\begin{claim}\label{claim:boolean_poly}
Let $\tau(x_1,\ldots,x_k)$ be a Boolean formula and $P_\tau(x_1,\ldots,x_k)$ its corresponding polynomial.
For every $c_1,\ldots,c_k \in \set{0,1}$ we have $P_\tau(c_1,\ldots,c_k) \in \set{0,1}$ and
$$
(c_1,\ldots,c_k) \models \tau \iff P_\tau(c_1,\ldots,c_k) = 1.
$$
\end{claim}
To ease the notation we will directly write formulas as polynomials; for instance,
$P_\tau(x,y) = (x\wedge y) \vee \neg y$.
All sequences in this section will be over $\set{0,1}$ and the involved polynomials will be of the form $P_\tau$.

\subparagraph{Sequences $\bc^1, \ldots, \bc^k$}
% Recall that $k$ is the number of variables in \eqref{eq:qbf}.
The truth valuations of variables $x_1, \dots, x_k$ will be encoded by sequences $\bc^1, \ldots, \bc^k$,
where for every $i$ and $n$ we have
\begin{align}\label{eq:ci}
\begin{split}
 c^i_n = \begin{cases}
          0 & \text{ if } n \bmod 2^{i} \text{ is less than } 2^{i-1},\\
          1 & \text{ otherwise.}
         \end{cases}
\end{split}
\end{align}
For example, the first eight values of $\bc^1,\bc^2,\bc^3$ are 

\vspace{0.2cm}
\begin{tabular}{ccccccccccc}
$\bc^1$ &=&0&1&0&1&0&1&0&1 \\
$\bc^2$ &=&0&0&1&1&0&0&1&1 \\
$\bc^3$ &=&0&0&0&0&1&1&1&1&.
\end{tabular}
\vspace{0.2cm}

\begin{claim}\label{claim:bc_size}
       For every $i \geq 1$,
       the sequence $\bc^i$ is definable by an extended polyrec system over $\bbF$ of size polynomial in $i$.
\end{claim}

\begin{claimproof}
        %First of all, the initial condition is $c^i_0 = 0$, for every $i$.
        %The recursive equations are as follows.
        We proceed by induction on $i$.
        For $i = 1$, by definition we have $c^1_n = 1 - c^1_{n-1}$,
        and thus we let
        \begin{align}\label{eq:c0}
        %c^1_0 = 0, \quad c^1_n = 1 - c^1_{n-1},
                c^1_n = P(c^1_{n-1}), \text{ with } P(x) = \neg x.
        \end{align}
        %and thus this recursion is defined by the polynomial $P(x) = \neg x$.

        Now, suppose $i > 1$ and we have defined $\bc^{i-1}$.
        We start by proving the following equality for every $n > 1$
        \begin{align}\label{eq:ck_if}
        c^i_n = \begin{cases}
                1 - c^{i}_{n-1} & \text{ if } c^{i-1}_{n-1} = 1 \text{ and } c^{i-1}_n = 0; \\
                c^{i}_{n-1} & \text{ otherwise.}
                \end{cases}
        \end{align}
        Notice that $\bc^i$ is periodic with period $2^i$, i.e., $c^i_{n} = c^i_{n + 2^i}$ for all $n$. Thus it suffices to prove \eqref{eq:ck_if} for $n \in \set{1,\ldots,2^i}$.
        By definition, $c^{i-1}_{n-1} = 1$ and $c^{i-1}_n = 0$ hold precisely for two values of $n \in \set{1,\ldots,2^i}$,
        namely for $n = 2^{i-1}$ and $n = 2^i$.
        Thus \eqref{eq:ck_if} is proved since $c^{i}_n = 0$ for $0 \le n < 2^{i-1}$; $c^{i}_n = 1$ for $2^{i-1} \le n < 2^i$; and $c^{i}_{2^i} = 0$.

        Using \eqref{eq:ck_if} one can determine $c^i_n$ given $c^i_{n-1}$, $c^{i-1}_{n-1}$, and~$c^{i-1}_{n}$:
        \begin{align}
                \label{eq:ck_poly}
                c^i_n = Q(c^i_{n-1}, c^{i-1}_{n-1}, c^{i-1}_n),
        \end{align}
        where $Q(x,y,z) = \left(\neg x \wedge (y \wedge \neg z)\right) \vee \left( x \wedge (\neg y \vee z) \right)$.
        This follows from \eqref{eq:ck_if} and from the fact that $y \wedge \neg z$ and $\neg y \vee z$ are mutually exclusive formulas encoding the ``if'' condition in \eqref{eq:ck_if}. It is clear that the constructed extended polyrec system is of size polynomial in $i$.
\end{claimproof}

\subparagraph{Sequences $\bd^0,\ldots, \bd^k$}

We define sequences $\bd^0,\ldots, \bd^k$,
where for any $i \geq 0$ we have:
\begin{align}\label{eq:di}
        \begin{split}
        d^i_0 = 0, \quad
        d^i_n = \begin{cases}
                0 & \text{ if } 2^{i} \not \divides \ n \\
                \sem{\psi|_{c^{i+1}_{n-1},\ldots,c^k_{n-1}}} & \text{ otherwise,}
                \end{cases}
        \end{split}
\end{align}
where for a closed formula $\xi$ (i.e., with no free variables) $\sem{\xi}$ is $1$ if $\xi$ is true and $0$ otherwise.
Notice that the formula depends on $\bc^1$, \ldots, $\bc^k$.
Since $\bd^k$ is the zero sequence if, and only if, $\psi$ is false,
it suffices to show that each $\bd^i$ can be defined by an extended polyrec system of polynomial size.

We proceed by induction on $i$.
In the base case $i = 0$,
\begin{align}
        d_n^0 = P_\varphi(c^1_{n-1},\ldots,c^k_{n-1}),
\end{align}
where $P_\varphi$ is the polynomial obtained from the quantifier-free formula $\varphi$ according to the rules in \eqref{eq:boolean}.
(Notice that $P_\varphi$ can be represented by an arithmetic circuit of size polynomial in the size of $\varphi$---this is where we use the conciseness of representation using circuits.)
This fulfills the conditions in~\eqref{eq:di} since,
for $n > 0$, $d_n^0 = 1$ if $(c^1_{n-1},\ldots,c^k_{n-1}) \models \varphi$
and $d_n^0 = 0$ otherwise.

Now, fix $i \ge 1$ and suppose that $\bd^{i-1}$ is defined. The goal is to define $\bd^i$. Recall that if $i$ is odd then $x_i$ is quantified existentially, and otherwise $x_i$ is quantified universally.

\begin{claim}\label{claim:psi}
%L: we do not need to introduce the new notation Q_i
%Consider the polynomial $Q_i(x,y) = x \circledast_i y$. Then
Let $\circledast_i = \vee$ if $i$ is odd and $\circledast_i =\wedge$ if $i$ is even.
For every $n > 0$ and $0 < i \le k$, we have
\begin{align}
 \begin{split}
d^i_n = \begin{cases}
         %Q_i(d^{i-1}_n, d^{i-1}_{n - 2^{i-1}})
         d^{i-1}_n \circledast_i d^{i-1}_{n - 2^{i-1}}
                & \text{ if } 2^i \divides n \\
         0 & \text{ otherwise.}
        \end{cases}  
 \end{split}
\end{align}

\end{claim}

\begin{claimproof}
        We may focus only on the case $2^i \divides n$.
        Since $x_i$ is quantified according to the parity of $i$, we have
        \begin{align*}
                \psi|_{c^{i+1}_{n-1},\ldots,c^k_{n-1}} = \psi|_{0, c^{i+1}_{n-1},\ldots,c^k_{n-1}} \; \circledast_i \; \psi|_{1, c^{i+1}_{n-1},\ldots,c^k_{n-1}}.
        \end{align*}

        We claim that
        \begin{align*}
                d^{i-1}_{n} = \psi|_{1, c^{i+1}_{n-1},\ldots,c^k_{n-1}}
                \ \text{ and }\ 
                d^{i-1}_{n - 2^{i-1}} = \psi|_{0, c^{i+1}_{n-1},\ldots,c^k_{n-1}}.
        \end{align*}
        By \eqref{eq:ci} and the fact that $2^i \divides n$,
        we get $c^i_{n-1} = c^i_{2^{i} - 1} = 1$, which proves the first equation. For the second equation, we observe that $c^{i}_{(n - 1) - 2^{i-1}} = c^{i}_{2^{i-1} - 1} = 0$ and that $c^{j}_{(n-1) - 2^{i-1}} = c^j_{n-1}$ for all $j > i$. The latter assertion  readily follows from $2^i\divides n$ and~\eqref{eq:ci}.
        %Indeed, by \eqref{eq:ci}, if $2^{i+1} \divides n$ then $c^j_{n-1} = c^{j}_{(n-1) - 2^{i-1}} = 1$,
        %otherwise $c^j_{n-1} = c^j_{2^{i} - 1} = 0$ and similarly $c^{j}_{(n-1) - 2^{i-1}} = 0$.
\end{claimproof}

As an immediate consequence of \cref{claim:psi}, we can write
\begin{align}\label{eq:polyS}
        d^i_n = S(d^{i-1}_n,d^{i-1}_{n - 2^{i-1}},c^{i-1}_{n-1},c^{i-1}_n),
\end{align}
where $S(x,y,z,t) = (x \circledast_i y) \wedge (z \wedge \neg t)$
(by recalling that $c^{i-1}_{n-1} = 1, c^{i-1}_n = 0$ holds if, and only if, $2^i \divides n$, where $n > 0$).
%
% By Claim~\ref{claim:psi} and \eqref{eq:di} for $n > 0$ we can write 
The issue with this recursive definition is that it requires access to the value $d^{i-1}_{n - 2^{i-1}}$,
which in general is not allowed in a polyrec system for $i \geq 2$ (not even in the extended variant).
This will be addressed in the next section by introducing the last family of recursive sequences.

\subparagraph{Sequences $\bff^0, \dots, \bff^{k-1}$}
For every $1 \leq i \leq k$,
the sequence $\bff^{i-1}$ is defined as
\begin{align*}
f_n^{i-1} = \begin{cases}
                0 & \text{ if } n \bmod 2^i \text{ is less than } 2^{i-1} \\
                d^{i-1}_{m} & \text{ otherwise,}
            \end{cases}
\end{align*}
where $m \le n$ is the unique number such that $n-m < 2^{i-1}$ and $2^{i-1} \divides m$. Thus, $\bff$ is divided into blocks of length $2^{i-1}$ of equal elements, where every other block is either filled with zeros, or its value is determined by the value of an appropriate entry $d^{i-1}_m$.
Observe that in particular, if $2^i \divides n$ then $f_{n-1}^{i-1} = d^{i-1}_{n - 2^{i-1}}$. Thus, intuitively, the sequence $\bff^{i-1}$ is a ``memory'' that allows us to store the relevant value of $\bd^{i-1}$ from $2^{i-1}-1$ steps back.

We now proceed to defining sequences $\bff^0,\ldots,\bff^{k-1}$ using polyrec systems. Observe that
$f_0^{i-1} = 0$ and for $n > 0$, we can write
\begin{align}\label{eq:fn}
        \begin{split}
        f_n^{i-1} =
                \begin{cases}
                        d^{i-1}_n & \text{ if } c^{i-1}_{n-1} = 0 \text{ and } c^{i-1}_n = 1; \\
                        0 & \text{ if } c^{i-1}_{n-1} = 1 \text{ and } c^{i-1}_n = 0; \\
                        f_{n-1}^{i-1} & \text{ otherwise.}
                \end{cases}
        \end{split}
\end{align}
Notice that the value of $f_n^{i-1}$ is copied from $f_{n-1}^{i-1}$
unless $c^{i-1}_{n-1}, c^{i-1}_{n}$ differ.
To conclude, recall from \eqref{eq:ci} that this happens if, and only if, $2^{i-1} \divides n$.

\begin{claim}\label{cl:final}
        For every $i \geq 1$,
        the sequences $\bd^i$ and $\bff^{i-1}$ are definable by extended polyrec systems over $\bbF$ of size polynomial in $i$ and the size of $\varphi$.
\end{claim}

\begin{claimproof}
        Using \eqref{eq:fn}, we may write $\bff^{i-1}$ as an extended polyrec sequence
        $f_0^{i-1} = 0$ and, for $n > 0$,
        \begin{align}\label{eq:fi}
                f^{i-1}_n = R(c^{i-1}_{n-1}, c^{i-1}_{n}, d^{i-1}_n, f^{i-1}_{n-1}),
        \end{align}
        where $$R(x,y,z,t) = (z \wedge (\neg x \wedge y)) \vee (t \wedge ((x \wedge y) \vee (\neg x \wedge \neg y)).$$
        In turn, this allows us to rewrite \eqref{eq:polyS} as
        \begin{align}\label{eq:di2}
                d^i_n = S(d^{i-1}_n,f^{i-1}_{n-1} ,c^{i-1}_{n-1},c^{i-1}_n),
        \end{align}
        where $S$ was defined in \eqref{eq:polyS}.
        Note that \eqref{eq:fi} and \eqref{eq:di2} are in the extended polyrec format
        provided that we write the equations for the $\bff^{i}$'s after the equations for the $\bd^{i}$'s,
        and the latter after the equations for $\bc^{i}$'s (in order to avoid creating a cyclic dependency).
        In other words, the final extended polyrec system consists of equations \eqref{eq:ck_poly},
        followed by~\eqref{eq:di2}, and followed by~\eqref{eq:fi}, where each set of equations is numbered naturally according to the indices of sequences.

        % It follows that $\bd^i$ is equivalent to $d^i_0 = 0$ and for $n > 0$
        % \begin{align}\label{eq:d_final}
        % \begin{split}
        % d^i_n = \begin{cases}
        %         d^{i-1}_n \circledast_i f^{i-1}_{n-1} & \text{ if } c^{i-1}_{n-1} = 1 \text{ and } c^{i-1}_n = 0\\
        %         0 & \text{ otherwise,}
        %         \end{cases}
        % \end{split}
        % \end{align}
        % where the operator ``$\circledast_i$''
        % %$Q_i$ is the polynomial from
        % has been defined in \cref{claim:psi}.
        % To see why \eqref{eq:d_final} holds,
        % notice that $c^{i-1}_{n-1} = 1, c^{i-1}_n = 0$ holds if, and only if, $2^i \divides n$ for $n > 0$.

        The involved polynomials $P_\varphi$, $R$ and $S$ are all of size polynomial in the input size when represented as arithmetic circuits ($R$ and $S$ are even of constant size),
        and we have a polynomial number of equations.
        Thus, the definition above is an extended polyrec system of polynomial size.
\end{claimproof}

As discussed, \cref{cl:final} finishes the proof of \cref{theorem:pspace}.

\medskip

In the end, we discuss the Skolem problem: given a sequence $\bu$ to determine whether there is $n$ such that $u_n = 0$. This problem was extensively studied for the class of linear recursive sequences (see \eg~\cite{OuaknineWorrell:SIGLOG:2015}). For linear recursive sequences it is open whether the Skolem problem is decidable, but only $\NPTIME$-hardness is known~\cite[Corollary 2.1]{BlondelPortier:LAA:2002}.
For polyrec sequences, decidability of the Skolem problem is also open, but we can improve the lower bound.

\begin{corollary}
The Skolem problem is $\PSPACE$-hard for polyrec sequences.
\end{corollary}

\begin{proof}
Notice that in the proof of \cref{theorem:pspace} we define a system of sequences over $\set{0,1}$. It remains to observe that for such sequences the zeroness problem and the Skolem problem reduce to each other. Indeed, the nonzeroness problem of a sequence $\bu$ over $\set{0,1}$ is equivalent to the Skolem problem of $\bv$ defined as $v_n = 1 - u_n$.
\end{proof}

We conclude this section by noting that the reduction from \QBF that we have presented
produces a polyrec sequence which is identically zero if and only if the first exponentially many initial values thereof are zero.
We are not aware of examples requiring longer witnesses of zeroness for polyrec sequences.